\documentclass[11pt]{article}

\usepackage{fullpage}

\usepackage[utf8]{inputenc}
\usepackage[english]{babel}
\usepackage[compact]{titlesec}

\usepackage{latexsym, amssymb, bbm}
\usepackage{amsmath}
\usepackage{amsthm}
\usepackage{verbatim}
\usepackage{enumitem}
\usepackage{mathtools}

\usepackage{graphicx}
\usepackage{nicefrac,xspace}
\usepackage{authblk}

\usepackage[dvipsnames]{xcolor}
\usepackage[breaklinks]{hyperref}
\hypersetup{colorlinks=true,
            citebordercolor={.6 .6 .6},linkbordercolor={.6 .6 .6},%
citecolor=blue,urlcolor=black,linkcolor=blue,pagecolor=black}
\usepackage[nameinlink]{cleveref}

\usepackage{algorithm}
\usepackage[noend]{algpseudocode}

\newtheorem{theorem}{Theorem}[section]
\newtheorem{lemma}[theorem]{Lemma}
\newtheorem{proposition}[theorem]{Proposition}
\newtheorem{corollary}[theorem]{Corollary}

\theoremstyle{definition}

\theoremstyle{remark}

\DeclarePairedDelimiter{\norm}{\|}{\|}
\DeclarePairedDelimiter{\braket}{ \langle}{ \rangle}

\newcommand{\bkt}[2]{\braket{#1, #2}}
\newcommand{\cut}[1]{}

\def\E{\mathbb{E}}

\def\L{\mathbb{L}}

\def\R{\mathbb{R}}

\def\g{\gamma}
\def\d{\delta}
\def\e{\varepsilon}

\def\th{\theta}

\def\la{\langle}
\def\ra{\rangle}
\newcommand{\ip}[1]{{\langle #1 \rangle}}

\def\Vol{\text{Vol}}
\def\dist{\text{dist}}
\def\cg{\text{cg}}
\def\st{\text{st}}

\def\sse{\subseteq}

\def\poly{\text{poly}}

\renewcommand{\emptyset}{\varnothing}
\newcommand{\final}{\text{final}}
\newcommand{\OPT}{\mathsf{opt}}

\newcommand{\initOneLiners}{%
    \setlength{\itemsep}{0pt}
    \setlength{\parsep }{0pt}
    \setlength{\topsep }{0pt}
}
\newenvironment{OneLiners}[1][\ensuremath{\bullet}]
    {\begin{list}
        {#1}
        {\initOneLiners}}
    {\end{list}}

\title{Chasing Convex Bodies with Linear Competitive Ratio}
\author[1]{C.J.~Argue  } 
\author[1]{Anupam Gupta}
\author[2]{Guru Guruganesh}
\author[1]{Ziye Tang}
\affil[1]{Carnegie Mellon University}
\affil[2]{Google Research}

\begin{document}

\maketitle 


\begin{abstract}
We study the problem of chasing convex bodies online: given a sequence
of convex bodies $K_t\sse \R^d$ the algorithm must respond with points
$x_t\in K_t$ in an online fashion (i.e., $x_t$ is chosen before
$K_{t+1}$ is revealed). The objective is to minimize the sum of distances between successive points in this sequence. Bubeck et al.\ (STOC 2019) gave a $2^{O(d)}$-competitive algorithm for this problem.
We give an algorithm that is $O(\min(d, \sqrt{d \log T}))$-competitive for any sequence of length $T$.
\end{abstract}

\section{Introduction}

In the problem of \emph{chasing convex bodies}, we control a
point in $\R^d$, initially at $x_0 = 0$. At each timestep $t$, the
adversary gives a closed convex sets $K_t \sse \R^d$ as its
\emph{request}, and we must respond with a point $x_t \in K_t$. (In
essence we are moving our point to lie within $K_t$.) This response
must be done in an \emph{online} fashion, namely $x_t$ must be chosen before
the next set $K_{t+1}$ is revealed. The goal is to minimize the
total movement cost: if the process goes on for $T$ steps, this cost is
\[ALG := \sum_{t=1}^T \|x_t - x_{t-1}\|.\] We focus on the
Euclidean norm; there are other results known for general norms---see,
e.g.,~\cite{bubeck2018chasing} for results and references. We work in
the framework of \emph{competitive analysis}; we compare the algorithm's
cost to the optimal cost $OPT$ incurred by an all-knowing adversary, and
the ratio between the two is called the \emph{competitive ratio}.

Friedman and Linial~\cite{FL93} introduced the problem of chasing
convex bodies as a generalization of several existing online
problems. They showed that no algorithm can achieve a competitive
ratio lower than $\sqrt{d}$, gave an algorithm for $d=2$, and posed
getting an $O(f(d))$-competitive algorithm for general $d$ as an open
problem. Over the years, competitive algorithms were given for special
cases; see, e.g.,~\cite{FL93,fujiwara2008online,Sitters,Anto}.
Motivated by
problems in data center scheduling, this problem was also studied
under the name of \emph{smoothed online convex optimization} (SOCO); see, e.g., ~\cite{chen2018smoothed,goel2019beyond}.

A particularly relevant line of
recent work is that for \emph{nested} convex body chasing, where the
request sequence satisfies $K_t \supseteq K_{t+1}$ for all $t$. For
this case, a series of works~\cite{Bansal,argue2019nearly} culminated
in Bubeck et al.~\cite{bubeck2018chasing} giving an
$O(\min(d, \sqrt{d \log T}))$-competitive algorithm based on the
classical notion of the Steiner point of convex bodies, and a more
sophisticated algorithm with a near-optimal competitive ratio of
$O(\sqrt{d\log d})$. Moreover, Bubeck et al.~\cite{BubeckLLS19-stoc}
built on the recursive centroid approach
from~\cite{argue2019nearly} to give the first $f(d)$-competitive algorithm for the general non-nested
case in general dimensions $d$, thereby resolving the
open problem from~\cite{FL93}. Their algorithm has competitive ratio
$f(d) = 2^{O(d)}$, and it substantially extends the recursive approach
of~\cite{argue2019nearly} for the nested case using a several new arguments. 
Our main result gives a simpler algorithm with an improved competitive ratio.

\begin{theorem}
  \label{thm:main}
  There is an $O(\min(d, \sqrt{d \log T}))$-competitive algorithm for
  the general convex body chasing problem in $d$-dimensional Euclidean
  space.
\end{theorem}

Our algorithm appears in \Cref{sec:algorithm}. It uses the Steiner-point based 
algorithm for the nested case~\cite{bubeck2018chasing} applied to a suitable
body $\Omega_t$. The body is defined using the concept of \emph{work functions}
commonly used in online algorithms~\cite{Sitters}. Hence, 
our result relies on bringing together the two lines of investigation
that had held out hope for the problem: ideas based on the work function,
and those based on convex geometry. In~\Cref{sec:efficient}, we give the details 
of an efficient implementation. The reader not interested in
these computational issues can focus only on \Cref{sec:algorithm}.

Friedman and Linial also defined a more general problem of convex
function chasing, as a direct generalization of \emph{metrical
  task systems}~\cite{DBLP:journals/jacm/BorodinLS92}. Here the request at time $t$ is a convex function
$g_t: \R^d \to \R_{\ge 0}$, and the algorithm must again respond with
a point $x_t \in \R^d$. The goal is to minimize
$\sum_t (g_t(x_t) + \| x_t - x_{t-1}\|)$, i.e., the total \emph{hit
  cost} in addition to the movement cost. Setting $g_t$ to be the
indicator of a convex set $K_t$ gives us back the problem of chasing
convex bodies. Bubeck et al.~\cite{BubeckLLS19-stoc} showed that
convex function chasing in $d$ dimensions can be reduced to convex
body chasing in $d+1$ dimensions using the epigraph. This reduction
loses only a constant factor in the competitive ratio. Hence we get
$O(\min(d, \sqrt{d \log T}))$-competitiveness for chasing convex
functions as well.
Independent and concurrent work by M.\ Sellke extends the notion of
Steiner points to convex functions in order to give a $\min(d, O(\sqrt{d
  \log T}))$-competitive
ratio for convex function chasing, and hence for convex bodies as well~\cite{Sellke}.

\subsection{Notation}
\label{sec:notation}

Given an instance of convex body chasing, the \emph{work
  function} $w_t(x)$ at time $t$ for the point $x \in \R^d$ is
the cost of the optimal trajectory that starts at $x_0$, satisfies the
first $t$ requests and ends at the point $x$. Formally, it is given by
the convex program:
\[\begin{array}{rl} 
w_t(x) = \text{min} 
	& \sum\limits_{s=1}^t \|x_s-x_{s-1}\| + \|x_t - x\| \vspace{2pt}\\
\text{s.t.} 
	& x_t\in K_t\\
	& x_0 = 0
\end{array} \]

Let $K\sse \R^d$ be a bounded convex body, and let $\cg(K)$ denote the
\emph{center of mass} of $K$. Let 
  $B(x_0,R):= \{x\in \R^d: \|x-x_0\| \le R\}$ denote the ball of radius $R$ 
  centered at $x_0$, and let $B := B(0,1)$ denote the unit ball in $\R^d$. 
The \emph{Steiner point} of the convex body $K$ is~\cite{Prz}:
\begin{gather}
  \st(K):= \lim_{s\to \infty} \cg(K+sB), \label{eq:1}
\end{gather}
where the sum is the Minkowski sum $K_1 + K_2 := \{ x_1+x_2 \mid x_i
\in K_i\}$.
The Steiner point has several other equivalent characterizations, including
these two~\cite[Section~5.4]{Sch-book}:
\begin{align}
  \st(K) &= d\int_{\theta\in S^{d-1}}\theta\, h_K(\theta) \,d\omega(\theta) =
d\,   \E_{\theta \sim \omega}[ \theta \, h_K(\theta)]
  \quad \text{and}
   \label{eq:2} \\
\st(K) &= \int_{\theta\in S^{d-1}} \nabla h_K(\theta) \, d\omega(\theta) =
\E_{\theta\sim \omega}[\nabla h_K(\theta)], \label{eq:st-def3}
\end{align}
where $\omega$ is the uniform (isometry invariant) measure
on $S^{d-1}$ with $\omega(S^{d-1}) = 1$,
the function \[ h_K(\theta) := \max_{x \in K} \ip{\theta, x} \] is the
\emph{support function} of the convex body $K$, and its gradient
\[ \nabla h_K(\theta) := \arg\max_{x \in K} \ip{\theta, x} \] is the
maximizing point (see \Cref{fig:grad_support} for an example). We will use 
definitions~(\ref{eq:1}) and~(\ref{eq:2}) in the two proofs for the 
algorithm, and~(\ref{eq:2}) for the efficient implementation.
Definition~(\ref{eq:st-def3}) shows that $\st(K)\in K$, which is crucial
to the proof of Bubeck et al. in \cite{bubeck2018chasing}, though we will not need this property 
for our proofs.

  \begin{figure}[t]
	\centering
	\begin{center}
		\includegraphics[height=1.5in,page=2]{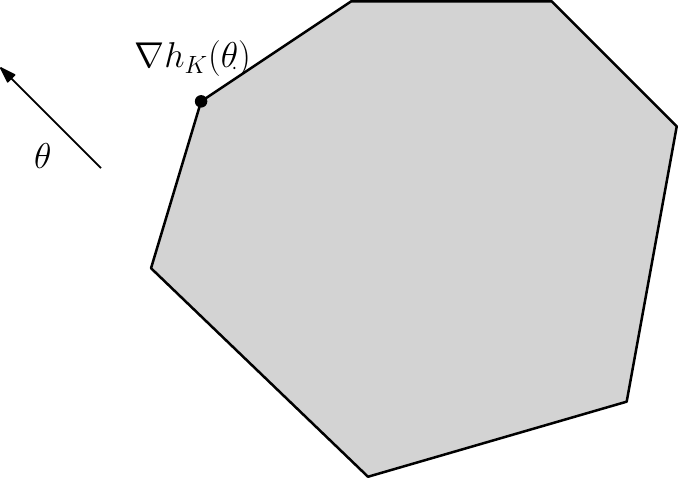} 
		\caption{$\nabla h_K(\theta)$ is the maximizer of $\max_{x\in K}\ip{\theta, x}$}
		\label{fig:grad_support}
	\end{center}
\end{figure}

See, e.g.,~\cite{BL-book,Sch-book}, for more about the
history and properties of the Steiner point. Among other facts, the
Steiner point is a \emph{selector} (i.e., a map from bounded convex bodies to
points inside them) with the smallest Lipschitz constant, with respect
to the Hausdorff distance on convex bodies. 
The work function has also been a staple in the online algorithms community. In fact, 
it achieves the best known deterministic competitive ratio for several problems
including the Metrical Task Systems Problem and K-Server Problem (see~\cite{BRS,KP95}). 

\section{The Algorithm}
\label{sec:algorithm}

The algorithm has an outer guess-and-double step, where
we maintain a current estimate $r$ that lies in $[OPT/2, OPT]$. Given
such an $r$, the algorithm is a single sentence: at each time $t$ we take
the $2r$-level set $\Omega_t$ of the work function $w_t(x)$, and move to its Steiner point.

\begin{algorithm}
  \caption{Steiner-Point-based General Chasing}
  \label{alg:algorithm-label}
  \begin{algorithmic}[1]
    \State $x_0 \gets 0$, $r \leftarrow \dist(x_0,K_1)$
    \Comment{Assume that $x_0\notin K_1$}
    \For{ $t = 1,\dots,$ }
    \State $\Omega_t \leftarrow \{x \mid w_t(x) \le 2r\}$
    \While{ $(\Omega_t = \emptyset)$}
    \State $r \leftarrow 2r$
    \State $\Omega_t \leftarrow \{x \mid w_t(x)\le 2r\}$     
    \EndWhile
    \State $x_t \leftarrow \st(\Omega_t)$ \label{alg:choice}
    \Comment{\Cref{lem:feasible} shows feasibility}
    \EndFor
  \end{algorithmic}
\end{algorithm}

\begin{lemma}
  \label{lem:convex}
  The work function $w_t$ is a convex function, and  $\Omega_t$
  is a bounded convex set.
\end{lemma}

\begin{proof}
  For $x, y \in \R^d$, let $\{x_s\}_{s=1}^t, \{y_s\}_{s=1}^t$ be the optimal
  solutions to the convex program that witness the values
  $w_t(x), w_t(y)$. Consider $z := \lambda x + (1-\lambda)y$. Then for
  each timestep $s \leq t$, define $z_s := \lambda x_s +
  (1-\lambda)y_s$ and note that $z_s \in K_s$ by convexity. Therefore,
  \[ w_t(z) \leq \sum_{s=1}^t\|z_s-z_{s-1}\| + \|z_t - z\| \leq \lambda w_t(x) +
    (1-\lambda)w_t(y), \]
  where the first inequality is because $\{z_s\}$ is a feasible
  solution to the convex program for point $z$ and $w_t(z)$ is the
  minimum value for it, and the second inequality is by the convexity of 
  the norm. Hence $w_t(\cdot)$ is
  a convex function, and its (sub-)level sets $\Omega_t$ are convex
  sets. Moreover, $\Omega_t \sse \{ x \mid \|x \| \leq 2r\}$ and hence
  is bounded.
\end{proof}

The all-important next claim shows that our choice of $x_t$ in Step~\ref{alg:choice} is feasible. \vspace{5pt}
\begin{lemma}
  \label{lem:feasible}
With $\Omega_t$ and $x_t$ as defined in the algorithm, $x_t\in K_t$.
\end{lemma}
\begin{proof}
  By the first definition~\eqref{eq:1} of $\st(\Omega_t)$, it suffices
  to show that for any fixed $s\ge 0$ we have $\cg(\Omega_t + sB) \in
  K_t$. To this end, we show that 
  $\cg(\Omega_t + sB)$ is contained in every halfspace containing $K_t$. 
  
  Let $H:=\{x\in \R^d \mid \la a, x\ra \ge b\}$ be a halfspace containing
  $K_t$ and let $H^=:= \{x\in \R^d \mid \la a, x\ra = b\}$ denote the supporting
  hyperplane of this halfspace. (Assume that $a$ is a unit vector.)
  For $x\not\in H$, define $\rho(x)$ to be the
  reflection of $x$ across $H^=$:
  \[ \rho(x) := x + 2(b - \la a, x\ra)\; a. 
  \]
  Consider that
  \[w_t(x) = \min_{y\in K_t} \left\{ \|x-y\| + w_{t-1}(y) \right\}.\]
  Let $y$ be the argmin of the expression on the right. Using that fact that
  $y\in K_t\sse H$, we have
  \[w_t(x) = \|x-y\| + w_{t-1}(y) \ge \|\rho(x) - y\| + w_{t-1}(y) \ge
    w_t(\rho(x)).\] It follows that if the point
  $x\in \Omega_t \setminus H$, then its reflection $\rho(x)\in \Omega_t$
  as well.

  \begin{figure}[t]
    \centering
    \begin{center}
      \includegraphics[height=1.5in,page=1]{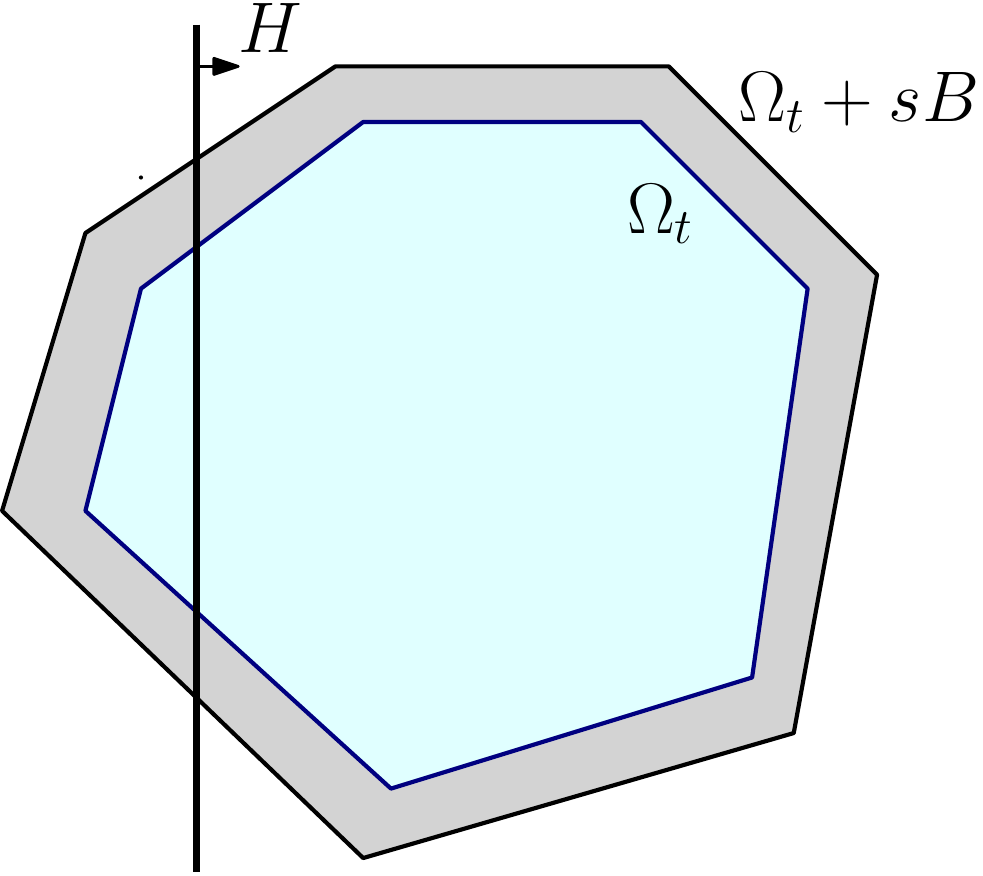} \qquad \qquad 
      \includegraphics[height=1.5in]{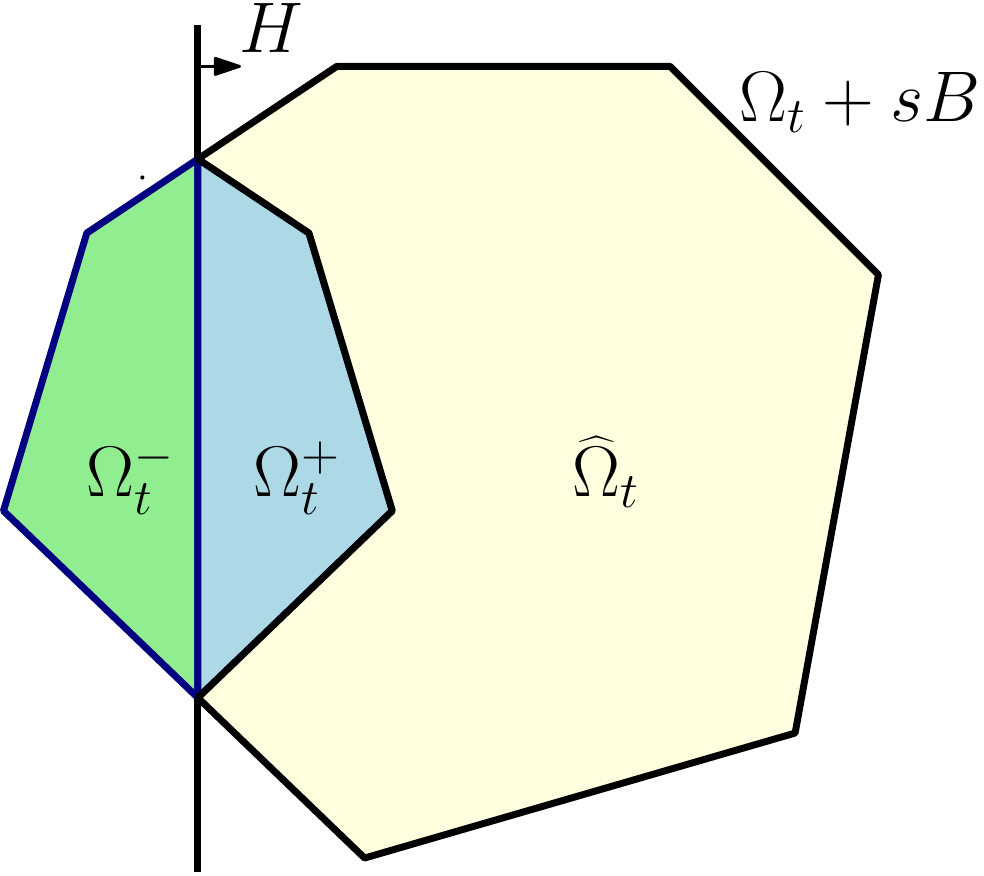}
      \caption{Proof of Lemma~\ref{lem:feasible}}
      \label{fig:slice}
    \end{center}
  \end{figure}

  A similar argument holds for any point
  $z\in (\Omega_t+ sB) \setminus H$: we claim that
  $\rho(z)\in \Omega_t+sB$. Get $x\in \Omega_t$ with $\|z-x\|\leq s$. If
  $x\in H$, then $\|\rho(z)-x\|\le \|z-x\|$ so $\rho(z)\in
  \Omega_t+sB$. Suppose $x\notin H$.  The above argument shows that
  $\rho(x)\in \Omega_t$ and furthermore
  $\|\rho(z)- \rho(x)\| = \|z-x\|\le s$, so $\rho(z)\in \Omega_t+sB$.
  Hence, the reflection of the part of $\Omega_t+sB$ that is infeasible
  for $H$ actually lies within $\Omega_t+sB$: namely, for any $s \geq 0$,
  \begin{gather}
    \rho((\Omega_t+sB)\setminus H)\sse (\Omega_t+sB)\cap H. \label{eq:3}
  \end{gather}

  For convenience, we will write the convex body 
  $\Omega_t + sB = \Omega^- \cup \Omega^+ \cup \widehat\Omega$
  where
  \begin{itemize}
  \item $\Omega^- := (\Omega_t+sB)\setminus H$.
  \item $\Omega^+ := \rho(\Omega^-)$.
  \item $\widehat \Omega := [(\Omega_t+sB)\cap H]\setminus \Omega^+$.
  \end{itemize}
  By symmetry $\cg(\Omega^- \cup \Omega^+)$ lies on $H^=$. 
  Since $\widehat\Omega\sse H$, we have $\cg(\widehat\Omega)\in H$.
  Letting $\g = \frac{\Vol(\Omega^+ \cup \Omega^-)}{\Vol(\Omega_t+sB)}$, 
  it follows that 
  \[\cg(\Omega_t+sB) 
    = \cg((\Omega^- \cup \Omega^+)\cup \widehat \Omega) = \gamma
    \cdot\cg(\Omega^- \cup \Omega^+) +
    (1-\gamma)\cdot\cg(\widehat\Omega) \in H.\] 
  But $H$ was chosen to be a generic halfspace containing $K_t$, so 
  $\cg(\Omega_t + sB)$ is contained in every halfspace $H$ containing
  $K_t$. 
\end{proof}

\begin{proof}[An Alternate Proof of \Cref{lem:feasible}] 
  This proof
  follows a reflection argument similar to the one in the first proof,
  but is based on definition~\eqref{eq:2} of the Steiner
  point. By translating the frame of reference, assume that the origin lies in
  $H^{=}$, namely $b=0$.

\begin{figure}[t]
    \centering
    \begin{center}
      \includegraphics[height=1.5in]{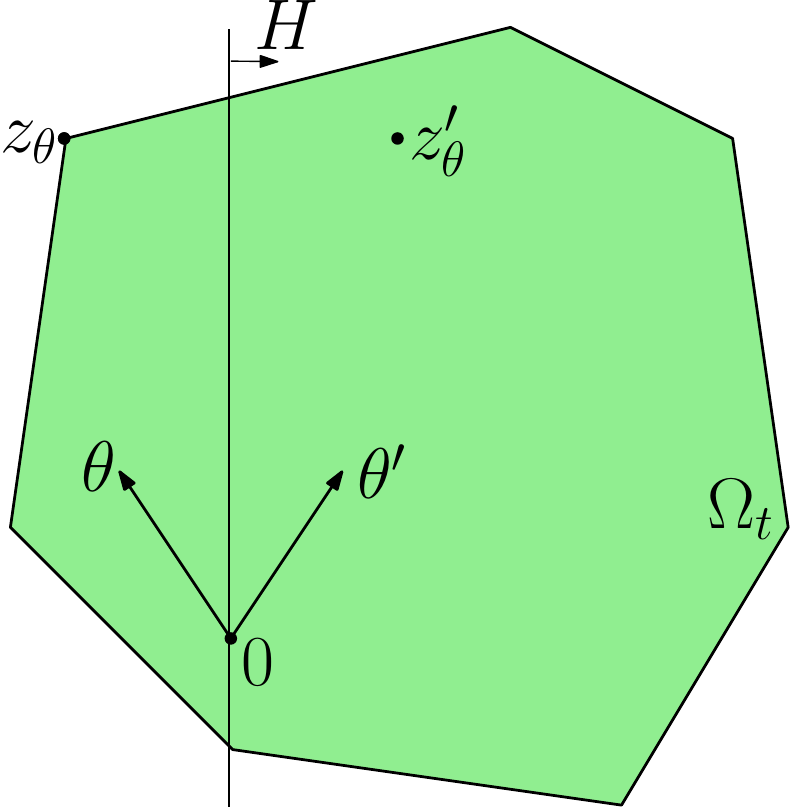}
      \caption{Alternate proof of Lemma~\ref{lem:feasible}}
      \label{fig:feas-lemma-2}
    \end{center}
\end{figure}
  
  Consider some direction $\theta \in S^{d-1}$
  such $\ip{\theta, a} \leq 0$ (i.e., $\theta$ lies to the left of the
  separating hyperplane in \Cref{fig:slice}).
  Let $z_\theta \in \Omega_t$ be such
  that $h_{\Omega_t}(\theta) = \ip{\theta, z_\theta}$.
  If $z_\theta \in H$ then define $z'_\theta = z_\theta$, else define
  $z'_\theta = \rho(z_\theta)$ to be the reflection of $z_\theta$ across
  $H^=$. By the arguments above, $z'_\theta \in \Omega_t$ as well, and
  moreover by construction $z'_\theta \in H$.  Finally, let
  $\theta' := \rho(\theta)$ be the reflection of
  $\theta$ across $H^=$. We claim that
  \begin{gather}
    \ip{z'_\theta, \theta'} \geq \ip{z_\theta, \theta} \label{eq:4}
  \end{gather}
  Indeed, we can rewrite~(\ref{eq:4}) as
  $\ip{z_\theta' - z_\theta, \theta'} + \ip{z_\theta, \theta' -
    \theta} \ge 0$. The first term is either 0 (if $z_\theta \in H$), or it is
  $-2\ip{z_\theta,a}\ip{a, \theta'} \geq 0$. To bound the second term, since
  $\theta'$ is the reflection of $\theta$ across $H^=$,
  $\ip{z_\theta, \theta' - \theta} = -2\ip{\theta, a}\ip{z_\theta, a}
  \geq 0$. This proves~(\ref{eq:4}), and we can infer that
  \[ h_{\Omega_t}(\theta') \geq \ip{z'_\theta, \theta'}
    \stackrel{(\ref{eq:4})}{\geq} \ip{z_\theta, \theta} =
    h_{\Omega_t}(\theta), \] where the first inequality follows from the
  definition of the support function as the maximizer.
So $\theta h_{\Omega_t}(\theta) + \theta' h_{\Omega_t}(\theta')$ is a
point in $H$.
Averaging over the choices of $\theta$ gives us a point also in
$H$. Since $H^=$ contains the origin, scaling by $d$ still gives a point in $H$.
\end{proof}

\Cref{thm:main} is now proved by applying standard doubling
arguments (as in \cite{argue2019nearly}) to the results of~\cite{bubeck2018chasing}.

\begin{proof}[Proof of \Cref{thm:main}]
  We consider the progression of the algorithm in \emph{phases}; each
  new phase begins when $r$ changes. Suppose the phase corresponding
  to some value of $r$ consists of times
  $\{t_1, \ldots, t_2\}$.  From the fact that the work function is
  non-decreasing over time, i.e., $w_t\le w_{t+1}$, it follows that
  \[\Omega_{t_1}\supseteq \Omega_{t_1+1} \supseteq \dots \supseteq \Omega_{t_2}.\]
  Moreover, $B(0,2r)\supseteq \Omega_t$ for each $t$ in this phase. 
  This gives an instance of the
  nested body chasing problem.  From Bubeck et
  al.~\cite[Theorem~3.3]{bubeck2018chasing}, the algorithm that moves to
  the Steiner point $\st(\Omega_t)$ at each time pays at most
  $(2r)\cdot O(\min(d, \sqrt{d \log T}))$. It follows that our algorithm
  pays at most that quantity plus an additive $O(r)$ during the phase,
  where the extra $O(r)$ comes from imaginarily moving to $0$ at the
  beginning of the phase.  Summing over all phases, we pay at most
    $r_{\final}\cdot O(\min(d, \sqrt{d \log T}))$, whereas $OPT$ pays at
    least $r_{\final}$. Hence the proof.
\end{proof}

\section{An Efficient Implementation for Chasing Half-Spaces}
\label{sec:efficient}

We now give an efficient implementation of our 
$O(\min(d,\sqrt{d\log T})$-competitive algorithm. For simplicity we consider the 
case where each body $K_t$ is a half-space
$\{x \mid a_t^\intercal x \geq b_t \}$. 
Henceforth, we will use $\L_t$ to denote
the bit-complexity of the sequence $\{(a_s, b_s)\}_{s \leq t}$.

The basic idea is the following: (a)~we show how to construct (weak)
evaluation oracles for the work function $w_t$ and for the support function
$h_{\Omega_t}$ of the body $\Omega_t$. Then (b)~using
Definition~(\ref{eq:2}) and the evaluation oracle for the support
function, we use random sampling to compute a point $x_t$ such that the
expected error
$\E[\|x_t - \st(\Omega_t)\|] \le O(\nicefrac{1}{t^2})\cdot OPT$; this
computation takes $\poly(\L_t)$ time. Now the total expected error over
all the time-steps is $O(\sum_t \nicefrac{1}{t^2})\cdot OPT = O(1) \cdot OPT$; this
adds only a constant to the competitive ratio. Along the way, we need to
change the algorithm slightly to control the bit-precision issues.

\subsection{The Modified Algorithm}
\label{sec:modified-algorithm}

To ensure that $\Omega_t$ remains full-dimensional, we change the
algorithm by stopping each phase when we know that $\OPT$ has increased
by a constant factor, but $\Omega_t$ still contains a large enough ball.

\begin{algorithm}
  \caption{Efficient Steiner Point Chasing}
  \label{alg:efficient-chasing}
  \begin{algorithmic}[1]
    \State $x_0 \gets 0$, $r \gets \dist(x_0, K_1)$.
    \For{$t = 1,\dots,$ }
       \If{$\textsc{MinWF}(w_t,\frac{r}{100})  > \frac{3r}{2} - \frac{r}{100}$} \label{ln:double}
         \State $r \gets \textsc{MinWF}(w_t,\frac{r}{100}) $ \label{ln:doub2}
       \EndIf
       \State $\Omega_t \gets \{x \mid w_t(x)\le 2r\}$
            \Comment $\Omega_t$ contains an $r/2$-ball \label{ln:choice-eff}
       \State $\widehat x_t \gets \textsf{Steiner}(\Omega_t,2r,\frac{r}{t^2}, 1)$  \label{ln:compute}
            \Comment Computed by \Cref{alg:appx-steiner-point} 

       \State $x_t \gets \textsc{Proj}(\widehat x_t, K_t)$ \label{ln:proj}
            \Comment $\textsc{Proj}(\widehat x_t, K_t) := \widehat x_t + (b_t - \la a_t, \widehat x_t\ra) \frac{a_t}{\|a_t\|}$
    \EndFor
  \end{algorithmic}
\end{algorithm}

Note the changes in the algorithm: when the minimum work-function value
(i.e., the optimal cost thus far) exceeds approximately $3r/2$, we reset
our estimate $r$ for the work function. This ensures that (a)~the
optimal cost increases by a constant factor each time we reset $r$, and
(b)~when we reach line~\ref{ln:choice-eff}, the optimal value is
guaranteed to be no more than $3r/2$, and hence $\Omega_t$ contains an
$r/2$ ball within it. The former property allows us to use the
``doubling trick'', and hence only the last phase matters. The latter is
useful for the computations, since $\Omega_t$ is ``centered'' (it is
sandwiched between two balls of $\Theta(r)$ radius).

There are two steps that need to be explained further. In~\Cref{cor:opt} we
show how to compute $\textsc{MinWF}(w_t, \e)$, the minimum work function
value up to additive error $\e$, which is needed in lines~\ref{ln:double}-\ref{ln:doub2}. In \Cref{sec:algo-steiner} we give a
randomized algorithm to compute an approximate Steiner point of the body
$\Omega_t$, needed in line~\ref{ln:compute}. This computation additionally requires the radius $R$ of a
bounding ball for $\Omega$ (which we set to $2r$), an error parameter
(which we set to $r/t^2$) and a failure probability parameter $\delta$
(which we can set to $1$ in our case since we are concerned only with expected cost).
\Cref{lem:steiner-appx} then ensures that the expected distance from the
true Steiner point is at most $\e_t = (1+\sqrt{\delta})r/t^2 = O(r/t^2)$.  
Finally, in line~\ref{ln:proj} we project this approximate Steiner point onto the 
half-space to get a feasible point.

\begin{lemma}  \label{lem:error-bound}
  Suppose the point $\widehat x_t$ satisfies
  $\E[\|\widehat x_t-\st(\Omega_t)\|] \le \e_t$ and $ALG$ outputs
  $x_t := \Pi_{K_t}(\widehat x_t)$, the projection of $\widehat{x}_t$
  onto the body $K_t$. Then the expected cost incurred by $ALG$ is at most
  $2\sum \e_t$ greater than the algorithm that plays the Steiner point
  $\st(\Omega_t)$ at each step.
\end{lemma}

 \begin{proof}
  Firstly, since $x_t$ is the projection of $\widehat x_t$ onto the body $K_t$, and
  $\st(\Omega_t) \in K_t$ 
  by~\Cref{lem:feasible}, we get $\|x_t - \st(\Omega_t)\| \le
  \|\widehat x_t - \st(\Omega_t)\|$. 
  Secondly, the triangle inequality implies:
  \begin{align*}
  \|x_t - x_{t-1}\|
  	&\le \|x_t - \st(\Omega_t)\| + \|\st(\Omega_t) - \st(\Omega_{t-1})\| 
     + \|\st(\Omega_{t-1}) - x_{t-1}\|.
     \intertext{Taking expectations, and using the assumptions of the lemma, we get
     }
\E[  \|x_t - x_{t-1}\|]  	&\le \e_t + \|\st(\Omega_t) - \st(\Omega_{t-1})\| 
	+ \e_{t-1}.
  \end{align*}
  Summing over all times completes the proof.
\end{proof}

\begin{corollary}\label{cor:alg2}
  The expected cost of
  \Cref{alg:efficient-chasing} is $O(\min(d,\sqrt{d \log
    T}))\cdot OPT$.
    Furthermore, the algorithm runs in time $\poly(\L_T)$. 
\end{corollary}
\begin{proof}
  The cost of the ideal algorithm that plays the actual Steiner point is
  $O(\min(d,\sqrt{d \log T}))\cdot OPT$. Consider $r_{\final}$, the
  final value of $r$ used by the algorithm; by construction
  $r_{\final} = O(1)\cdot OPT$. Since we chose
  $\e_t \leq O(1/t^2) r_{\final}$, the expected extra cost is
  $2\sum_t \e_t = O(1)\cdot OPT$.

  At step $t$, we can compute both $\textsc{MinWF}$ and $\textsc{Steiner}$ in 
  $\poly(\L_t, \log (1/\e_t))$.  Note that $\L_t \le \L_T$ and
  $\log(1 / \e_T) = \log (\poly(T))\le \poly(\L_T)$. 
  The computation time for step $t$ is at most $\poly(\L_T)$, and we perform at most $2T$ 
  such computations, hence the total computation time is bounded by $\poly(\L_T)$. 
\end{proof}

Next we show how to compute the work function, and the Steiner point.

\subsection{Computing the work and support functions}
\label{sec:weak-eval-oracle}

In this section, we show how to compute the minimum value of the  
work function $w_t$ and the support function $h_{\Omega_t}$ of the body $\Omega_t$.
The latter will be used in our estimation of the Steiner point
in~\Cref{sec:algo-steiner}.
We do each calculation by writing it as a second-order cone program.

\begin{lemma}[Minimization of Work Function]
    \label{cor:opt}
  Given a sequence $\{(a_i, b_i)\}_{i =1}^t$
  and an error tolerance $\e > 0$,
  there is a procedure $\textsc{MinWF}(w_t,\e)$  that outputs a
  value $v$ satisfying $|\min_x w_t(x) - v| \leq \e$ in time
  $\poly(\L_t,\log(1/\e))$. 
\end{lemma}
\begin{proof} 
  Consider the following program on 
  variables $(x_1, \lambda_1), \ldots, (x_t, \lambda_t)$, where each $\lambda_i$
  is a scalar (and recall that $x_0 = 0$):
  \begin{alignat*}{2}
    \min \quad \sum_{i=1}^t & \ \lambda_i  && \\
    \text{s.t.} \quad \norm{x_i - x_{i-1}} &\leq \lambda_i & \qquad \qquad & \forall i \in
    [t]\\
    \bkt{a_i}{x_i} &\geq b_i &&\forall i \in [t]
    \end{alignat*}
    The above program is a \emph{second-order cone program}. 
    (Second-order cone programming is a special case of semidefinite programming.)
    By setting $x_i = (b_i + \delta)
    \frac{a_i}{\norm{a_i}}$, we get that
    $\bkt{a_i}{x_i} \geq b_i +
    \delta$, and hence we get a point in the strict interior of the
    polyhedron $K_1 \times K_2 \times \ldots \times K_t$. 
    Since there is a strictly feasible point,  we know that the optimal primal
     and dual solutions exist, and have equal objective by~\cite[Thm 14]{AG03}.
     Furthermore, we know that the optimal solution has bit-complexity bounded by $\poly(\L_t)$. Therefore we can solve this problem, e.g., using interior-point methods, to get an
    $\e$-approximate solution with run-time
    $\poly(\L_t,\log(1/\e))$(see~\cite[Section~4.6.2]{BN2001}). 
\end{proof}

\begin{proposition}[Support Function Oracle]
  \label{fact:separation-oracle}
  Given a sequence $\{(a_i,b_i)\}_{i=1}^t$ and a value $r > 0$ (as defined in line~\ref{ln:doub2}), then
  we can implement an oracle $\textsc{Supp}_{\Omega_t}$ for the support
  function $h_{\Omega_t}$ of $\Omega_t$    such that
  $\textsc{Supp}_{\Omega_t}(\theta, \e)$ outputs 
    a value $v$ within $\e$ of $h_{\Omega_t}(\theta)$ and runs in time 
    $\poly(\L_t,  \log(\nicefrac{1}{\e}))$.  
\end{proposition}
\begin{proof}
We use a second-order cone program similar to the one in the previous proof.
We add variables $x$ and $\lambda$ corresponding to an extra step after satisfying
request $K_t$. Furthermore, we stipulate that the total distance traveled be at most $2r$
to ensure that $x\in \Omega_t$.
  \begin{alignat*}{2}
    \max  \la \th, x\ra & & & \\
    \text{s.t.} \quad \norm{x_i - x_{i-1}} &\leq \lambda_i & \qquad \qquad & \forall i \in [t]\\
     \norm{x- x_t} &\leq \lambda & \qquad \qquad & \forall i \in [t]\\
     \bkt{a_i}{x_i} &\geq b_i &&\forall i \in [t]\\
    \lambda + \sum_{i=1}^t \lambda_i &\leq 2r && 
    \end{alignat*}
    We will first show that the feasible set contains a point in the strict 
    interior. Let $x_t^*$ denote the minimizer of the work function $w_t$, and 
    let $(x_1^*, \dots, x_t^*)$ be an optimal path witnessing the value of 
    $w_t(x_t^*)$. We know that $\Omega_t$ contains the ball of radius $r/2$ 
    around $x_t^*$. Consider the points 
    $x_i = x_i^* + \left(\frac{r}{10\cdot t^2}\right) \frac{a_i}{\norm{a_i}}$ 
    and set $x = x_t$. By construction, we know that each $\bkt{a_i}{x_i} > b_i$. 
    Furthermore, we get a point that satisfies the distance constraint strictly:
    \begin{align*}
        \sum_{i=1}^{T} \norm{x_i - x_{i-1}} + \norm{x_t - x} 
        	&= \sum_{i=1}^{T} \norm{x_i - x_{i-1}} \\
			&\leq \sum_{i=1}^T \left( \norm{x_i^* - x_{i-1}^*} 
				+ 2 \cdot \frac{r}{10 t^2} \right) 
			\leq \frac{3r}{2} +  \sum_{i=1}^T  \frac{r}{5 t^2}  < 2r.
    \end{align*}
    The second inequality follows from the triangle inequality, and the third
    inequality comes from the early doubling in line~3. 
    Since the optimal primal solution is bounded above by $2r$ and there is a
    strictly feasible point, we know that the optimal primal and dual solutions
    exist, and have equal objective (again, by~\cite[Thm 14]{AG03}). We can now
    use interior point methods to get an $\e$-approximate solution in time 
    $\poly(\L_t,\log( \nicefrac{1}{\e}))$ as in the preceding proof.
\end{proof}

\subsection{Approximating the Steiner Point}
\label{sec:algo-steiner}

Finally, we show how to approximate the Steiner point of a convex body using
the oracle for the support function. Recall the
definition~(\ref{eq:2}) of the Steiner point $\st(K)$ as
\begin{gather*}
  \st(K) = d\; \E_{\theta \sim S^{d-1}} [ \theta\, h_K(\theta) ], 
\end{gather*}
where $h_K(\theta):= \max_{x\in K} \bkt{ \theta}{x} $ is the support function
of the convex body $K$.

We are given half spaces $\{(a_1,b_1),\dots,(a_t,b_t)\} $ and access to an oracle
 $\textsc{Supp}_{K}(\theta, \e)$ (defined
in \Cref{fact:separation-oracle}) to compute the support function; given a vector 
$\theta \in S^{d-1}$, this oracle returns a value $v$ such that
$|v - h_K(\theta)| \leq \e$. We assume that we are also given a radius
$R$ of a bounding ball $B(0,R)\supseteq K$. The estimation algorithm 
samples some number $N$ of random directions $\theta$, uses
the oracle to get estimates for $h_K(\theta)$, and then uses the
empirical mean as an estimate for the expectation
in~\eqref{eq:2}. Formally, the algorithm is specified as
\Cref{alg:appx-steiner-point}.

\begin{algorithm}
  \caption{\textsf{Steiner}($K,R, \e, \delta$)}
  \label{alg:appx-steiner-point}
  \begin{algorithmic}[1]
      \State $N \leftarrow \frac{(d+1)^2 R^2}{\e^2 \delta}$
    \For{ $i = 1,\dots, N$ }
    \State Sample $\theta_i$ uniformly from $S^{d-1}$
    \State $p_i \gets \textsc{Supp}_{K}(\theta_i,\e/d)$
    \EndFor
    \State $S \gets \frac{d}{N}\sum_{i=1}^N p_i\th_i$
    \State \Return $S$
  \end{algorithmic}
\end{algorithm}

We show that~\Cref{alg:appx-steiner-point} computes a $2\e$-approximate
Steiner point both in expectation and with probability $1-\delta$, although 
only the former is necessary for this work.

\begin{lemma} \label{lem:steiner-appx} For $\e, \delta > 0$, suppose
  $K\sse B(0,R)$ is a bounded convex body specified by a weak linear
  optimization oracle $\textsc{Supp}_{K}$. Let $\hat{x}$ be the output of
  $\textsf{Steiner}(K,R,\e,\d)$. Then 
  \begin{OneLiners}
  \item[(a)] $\E[\| \hat{x}-\st(K)\|] \leq \e (1+\sqrt{\delta})$, and
  \item[(b)] $\Pr[\| \hat{x}-\st(K)\| \leq 2\e] \geq 1-\d$.
  \end{OneLiners}
\end{lemma} 

\begin{proof}
  Let $x:= \st(K)$. 
  We first ignore the error of the linear optimization oracle and
  bound only the error due to sampling. Let $q_i := h_K(\th_i)$ be the \emph{true}
  maximum value in the direction of the unit vector $\th_i$ and let
  $\widehat S_N = \frac{d}{N}\sum_{i=1}^N q_i\th_i$. Since the
  $(\th_i,q_i)$ are i.i.d.\ samples with $\E[dq_i\theta_i] = \st(K) = x$  and
  $q_i\le \max_{z\in B(0,R)} \la \th_i, z\ra = R$, we have
  \[\E\left[\|\widehat S_N - \st(K)\|^2\right] 
  	= \E\left[ \left\| \frac1N\sum_{i=1}^N(dq_i\theta_i- x) \right\|^2\right] 
	= \frac{1}{N^2} \sum_{i =1}^N \E\left[\left\| dq_i\th_i - x\right\|^2\right]
  	\leq \frac{(d+1)^2 R^2}{N} \leq \e^2 \delta.
  \]
  Now Jensen's inequality says $\E[\|\widehat S_N -
  \st(K)\|]^2 \leq \E[\|\widehat S_N - \st(K)\|^2] \leq \e^2 \delta$. 
  The second-last inequality follows since both $q_i\theta_i $ and $x$ lie in $B(0,R)$. 
  Moreover, using Markov's inequality in conjunction with the second inequality also
  implies that $\|\widehat S_N - \st(K)\| \le \e$ with
  probability at least $1-\delta$.
  
  Finally, to account for error in the weak linear optimization oracle,
  note that the magnitude of each error is $|p_i - q_i|< \e/d$, hence the
  convexity of norms implies that the additional error between the
  algorithm's output $S$ and the true estimate $\widehat{S}_N$ is 
  $\|S - \widehat S_N\| \le \e$.
  The triangle inequality completes the proof.
\end{proof}

\subsection*{Acknowledgments}

We thank Nikhil Bansal, S\'ebastien Bubeck, Niv Buchbinder, Marco
Molinaro, Seffi Naor, Kirk Pruhs, and Cliff Stein for comments and discussions.

{\small
\bibliographystyle{alpha}
\bibliography{arxiv2}}

\end{document}